\documentclass{amsart}\usepackage{amsmath}

\usepackage{amssymb}
\usepackage{graphicx}
\usepackage{amscd}

\newtheorem{theorem}{Theorem}
\theoremstyle{plain}

\newtheorem{lemma}{Lemma}

\numberwithin{equation}{section}

\begin{document}
\title[On the GPI ]{The Generalized Poverty Index}
\author{Gane Samb LO}
\address{LERSTAD, Universit\'{e} Gaston Berger, BP 234 Saint-Louis, S\'{e}n%
\'{e}gal.}
\address{Affiliated to LSTA, Universt\'e Paris VI, UPMC, France.}
\email{gslo@ufrsat.org}
\urladdr{http://www.ufrsat.org/perso/gslo}

\begin{abstract}
{\large We introduce the General Poverty Index (GPI), which summarizes most
of the known and available poverty indices, in the form 
\begin{equation*}
GPI=\delta (\frac{A(Q_{n},n,Z)}{nB(Q,n)}\overset{Q_{n}}{\underset{j=1}{\sum }%
}w(\mu _{1}n+\mu _{2}Q_{n}-\mu _{3}j+\mu _{4})d\left( \frac{Z-Y_{j,n}}{Z}%
\right) ),
\end{equation*}
where 
\begin{equation*}
B(Q_{n},n)=\sum_{j=1}^{Q}w(j),
\end{equation*}
$A\left( \cdot \right) ,$ $w\left( \cdot \right) ,and$ $d\left( \cdot
\right) $\ are given measurable functions, $Q_{n}$ is the number of the poor
in the sample, Z is the poverty line and $Y_{1,n}\leq Y_{2,n}\leq ...\leq
Y_{n,n}$\ are the ordered sampled incomes or expenditures of the individuals
or households. We show here how the available indices based on the poverty
gaps are derived from it. The asymptotic normality is then established and
particularized for the usual poverty measures for immediate applications to
poor countries data. }
\end{abstract}

\maketitle
\Large
\section{Introduction}

The Economists are interested in monitoring the welfare of the worse-off in
one given population. In this capacity, poverty measures are defined and used to
compare subgroups and to follow the evolution of the poor with respect to
time. A poverty measure is assumed to fulfill a number of axioms since the
pioneering work of Sen (\cite{sen}). Many authors proposed poverty indices
and studied their advantages, like Sen himself, Thon (\cite{thon}), Kakwani (%
\cite{kakwani}), Clark-Hemming-Ulph ( \cite{clark}), Foster-Greer-Thorbecke (%
\cite{fgt}), Ray (\cite{ray}), Shorrocks (\cite{shorrocks}). Most of the
required properties for such indices are stated and described in (\cite%
{zheng}) along with a broad survey of the available poverty indices. 

Asymptotic theories for theses quantities, when they come from random
samplings, have been given in recent years. Dia(\cite{dia}) used point
process theory to give asymptotic normality for the Foster-Greer-Thorbecke
(FGT) index. Sall and Lo(\cite{lo1}) studied an asymptotic theory for the
poverty intensity defined below and further, Sall, Seck and Lo(\cite{lo2})
proved a larger asymptotic normality for a general measure including the
Sen, Kakwani, FGT and Shorrocks ones. 

Now our aim, here, is to unify the monetary poverty measurements with
respect as well to Sen's axiomatic approach as to the asymptotic aspects. We
point out that poverty may be studied through aspects other than monetary
ones as well. It can be viewed through the capabilities to meet basic needs
(food, education, health, clothings, lodgings, etc.). In our monetary frame,
the main tools are the poverty indices. We give, here, a general poverty
index denoted as the General Poverty Index (GPI), which is aimed to
summarize all the known and former ones. Let us make some notation in order
to define it. 

We consider a population of individuals, each of which having a random
income or expenditure $Y$\ with distribution function $G(y)=P(Y\leq y).$\ An
individual is classified as poor\ whenever his income or expenditure $Y$\
fulfills $Y<Z,$\ where $Z$\ is a specified threshold level (the poverty
line). 

Consider now a random sample $Y_{1},Y_{2},...Y_{n}$\ of size $n$\ of
incomes, with empirical distribution function $G_{n}(y)=n^{-1}\#\left\{
Y_{i}\leq y:1\leq i\leq 1\right\} $. The number of poor individuals within
the sub-population is then equal to $Q_{n}=nG_{n}(Z).$ 

Given these preliminaries, we introduce measurable functions $A(p,q,z)$, $%
w(t)$, and $d(t)$\ of $p,q\in N,$\ and $z,t\in R$. The meaning of these
functions will be discussed later on. Set $B(Q_{n},n)=\sum_{i=1}^{q}w(i).$ 

Let now $Y_{1,n}\leq Y_{2,n}\leq ...\leq Y_{n,n}$\ be the order statistics
of the sample $Y_{1},Y_{2},...Y_{n}$\ of $Y$. We consider general\ \
poverty\ indices\ \ of the form 

\begin{equation}
GPI_{n}=\frac{A(Q_{n},n,,Z)}{nB(Q_{n},n)}\sum_{j=1}^{Q_{n}}w(\mu _{1}n+\mu
_{2}Q_{n}-\mu _{3}j+\mu _{4})\text{\ }d\left( \frac{Z-Y_{j,n}}{Z}\right) ,
\label{gpi01}
\end{equation}%
where $\mu _{1},\mu _{2},\mu _{3},\mu _{4}$\ are constants. By
particularizing the functions $A$ and $w$ and by giving fixed values to the $%
\mu _{i}^{\prime }s,$ we may find almost all the available indices, as we
will do it later on. In the sequel, (\ref{gpi01}) will be called a poverty
index (indices in the plural) or simply a poverty measure according to the
economists terminology. 

The poverty line Z is defined by economics specialists or governmental
authorities so that any individual or household with income (say yearly)
less than Z is considered as a poor one. The poverty line determination
raises very difficult questions as mentioned and shown in ( \cite{kak2}).
We suppose here that Z is known, given and justified by the specialists.
\bigskip 

Our unified and global approach will permit various research works, as well
in the Statistical Mathematics field as in the Economics one. It happens
that poverty indices are also somewhat closely connected with economic
growth questions. We should find conditions on the functions and the
constants in (\ref{gpi01}) so that any kind of needed requirements are met
and that the hypotheses imposed by the asymptotic normality are also
fulfilled. This may lead to a class of perfect or almost perfect poverty
measures. In this paper, we concentrate on the description of the GPI and on
the asymptotic normality theory. Our best achievement is that (\ref{gpi01}),
is asymptotically normal for a very broad class of underlying distributions.
These results are then specialized for the particular and popular indices. 

We then begin to describe all the available indices in the frame of (\ref%
{gpi01}) in the next section. In section 3, we establish the asymptotic
normality. Related application works to poverty databases can be found in 
\cite{lo6} for instance. 

\section{How does the GPI include the poverty indices}

We begin by making two remarks. First, for almost all the indices, the
function $\delta (\cdot )$ is the identity one 
\begin{equation*}
\forall (u\geq 0),\text{ }\delta (u)=I_{d}(u)=u.
\end{equation*}

We only noticed one exception in the Clark-Hemming-Ulph (CHUT) index.
Secondly, we may divide the poverty indices into non-weighted and weighted
ones. The non weighted measures correspond to those for which the weight is
constant and equal to one : 
\begin{equation*}
w(\mu _{1}n+\mu _{2}Q_{n}-\mu _{3}j+\mu _{4})\equiv 1.
\end{equation*}%
We begin with them. 

\subsection{The non-weighted indices.}

First of all, the Foster-Greer-Thorbecke (FGT) index of parameter \cite{fgt}
defined for $\alpha \geq 0,$%
\begin{equation}
FGT_{n}(\alpha )=\frac{1}{n}\overset{Q_{n}}{\underset{j=1}{\sum }}\left( 
\frac{Z-Y_{j,n}}{Z}\right) ^{\alpha }.  \label{sall2}
\end{equation}

is obtained from the GPI with 
\begin{equation*}
\begin{array}{cccc}
\delta =I_{d}, & w\equiv 1, & d(u)=u^{\alpha }, & B(Q_{n},n)=Q_{n}\text{ }and%
\text{ }A(Q_{n},n,Z)=Q_{n}.%
\end{array}%
\end{equation*}%
The Ray index defined by (see \cite{ray}), for $\alpha>0$,  
\begin{equation}
R_{R,n}=\frac{g}{nZ}\sum_{i=1}^{Q_{n}}((Z-Y_{j,n})/g)^{\alpha }
\end{equation}%
where  
\begin{equation}
g=\frac{1}{Q_{n}}\sum_{j=1}^{j=Q_{n}}(Z-Y_{j,n})
\end{equation}%
is derived from the GPI with  
\begin{equation*}
\begin{array}{cccc}
\delta =I_{d}, & w\equiv 1, & d(u)=u^{\alpha }, & B(Q_{n},n)=Q_{n}\text{ }and%
\text{ }A(Q_{n},n,Z)=Q_{n}(g/Z)^{\alpha -1}.%
\end{array}%
\end{equation*}%
The coefficient $A(Q_{n},n,Z)$\ depends here on the income or the
expenditure. This is quite an exception among the poverty indices. We may
also cite here the Watts index \ (see \cite{watts})  
\begin{equation*}
P_{W,n}=\frac{1}{n}\sum_{j=1}^{j=Q_{n}}(\ln Z-lnY_{j,n}).
\end{equation*}%
But this may be derived from the FGT one as follows. The income Y is
transformed into $\ln Y$\ and, consequently, the poverty line is taken as $%
\ln Z.$\ \ It follows that  
\begin{equation*}
W(Y)=FGT(1,\ln Y)
\end{equation*}%
for the poverty line $\ln Z$. The case is similar for the Chakravarty index
(see \cite{chakravarty}), $0<\alpha <1,$\  
\begin{equation*}
P_{Ch}=\frac{1}{n}\sum_{j=1}^{j=Q_{n}}(1-(\frac{Y_{j,n}}{Z})^{\alpha }).
\end{equation*}%
We may consider it through the FGT class  
\begin{equation*}
W(Y)=FGT(1,Y^{\alpha })
\end{equation*}%
for the poverty line $Z^{\alpha}$.

\noindent Now, we have that the CHU index is clearly of the GPI form with  

\begin{equation*}
\begin{array}{cccc}
\delta =(u)=u^{1/\alpha }, & w\equiv 1, & d(u)=u^{\alpha } & B(Q_{n},n)=Q_{n}%
\text{ }and\text{ }A(Q_{n},n,Z)=Q_{n}^{\alpha }/n^{\alpha -1}.%
\end{array}%
\end{equation*}%

\noindent Now let us describe the weighted indices.

\subsection{The weighted indices}

First, the Kakwani (\cite{kakwani}) class of poverty measures  
\begin{equation}
P_{KAK,n}(k)=\frac{Q_{n}}{n\Phi _{k}(Q_{n})}\overset{Q_{n}}{\underset{j=1}{%
\sum }}(Q_{n}-j+1)^{k}\left( \frac{Z-Y_{j,n}}{Z}\right) ,  \label{sall4}
\end{equation}%
where  
\begin{equation*}
\Phi _{k}(Q_{n})=\sum_{j=1}^{j=Q_{n}}j^{k\text{ \ }}=B(Q_{n},n)
\end{equation*}%
comes from the GPI with  
\begin{equation*}
\delta =I_{d},\text{ }w(u)\equiv (u),\text{ }d(u)=u,\text{ }\mu _{1}=0,
\end{equation*}%
\begin{equation*}
\mu _{2}=1,\text{ }\mu _{3}=-1,\text{ }\mu _{4}=1\text{ }and\text{ }%
A(n,Q_{n},Z)=Q_{n}
\end{equation*}%
For $k=1,$\ $P_{KAK,n}(1)$\ is nothing else but the Sen poverty measure  
\begin{equation}
P_{Sen}=\frac{2}{n(Q_{n}+1)}\overset{Q_{n}}{\underset{j=1}{\sum }}%
(Q_{n}-j+1)\left( \frac{Z-Y_{j,n}}{Z}\right) .  \label{sall3}
\end{equation}%
As to the Shorrocks (\cite{shorrocks}) index  
\begin{equation}
P_{SH,n}=\frac{1}{n^{2}}\overset{Q_{n}}{\underset{j=1}{\sum }}%
(2n-2j+1)\left( \frac{Z-Y_{j,n}}{Z}\right) ,  \label{sall5}
\end{equation}%
it is obtained from the GPI with 
\begin{equation*}
B(Q_{n},n)=Q(2n-Q),\text{ }A(n,Q_{n},Z)=Q_{n}(2n-Q_{n})/n
\end{equation*}%
and\  
\begin{equation*}
\begin{array}{cccccc}
\delta =I_{d}, & w(u)\equiv (u), & d(u)=u, & \mu _{1}=2, & \mu _{2}=0, & \mu
_{3}=2,%
\end{array}%
\mu _{3}=1.
\end{equation*}%
Thon (\cite{thon}) proposed\bigskip\ the following measure  
\begin{equation*}
P_{Th}=\frac{2}{n(n+1)}\overset{Q_{n}}{\underset{j=1}{\sum }}(n-j+1)\left( 
\frac{Z-Y_{j,n}}{Z}\right) 
\end{equation*}%
which belongs to the GPI family for  
\begin{equation*}
B(n,Q_{n})=Q_{n}(n-Q_{n}+1)/2,\text{ }A(n,Q_{n},Z)=Q(n-Q+1)/(n+1),
\end{equation*}%
and 
\begin{equation*}
\begin{array}{cccccc}
\delta =I_{d}, & w(u)\equiv u, & d(u)=u, & \mu _{1}=1, & \mu _{2}=0, & \mu
_{3}=1,%
\end{array}%
\begin{array}{cc}
\mu _{3}=1. & 
\end{array}%
\end{equation*}%

\noindent Not all the poverty indices are derived from the GPI. What precedes only
concerns those based on the poverty gaps  
\begin{equation*}
(Z-Y_{j}),\text{ }1\leq j\leq Q_{n}.
\end{equation*}%
We mention one of them in the next subsection. 

\subsection{An index not derived from the GPI}

The Takayama (\cite{takayama}) index  
\begin{equation*}
P_{TA,n}=1+\frac{1}{n}-\frac{2}{\mu n^{2}}\sum_{j=1}^{Q_{n}}(n-j+1)Y_{j,n},
\end{equation*}%
where $\mu $\ is the empirical mean of the censored income, cannot be
derived from the GPI. The main reason is that, it is not based on the
poverty gaps $Z-Y_{j,n}.$\ It violates the monotonicity axiom which states
that the poverty measure increases when one poor individual or household
becomes richer. 

$\bigskip $ 

Now we must study the so-called GPI with respect to the axiomatic approach
as well as to the asymptotic theory. We focus in this paper to the general
theory of asymptotic normality. The interest of this unified approach is
based on the fact that we obtain at once the asymptotic behaviors for all
the available poverty indices, as particular cases. Indeed, in the next
section, we will describe apply the general theorem to the particular usual
indices. 

\section{Asymptotic normality of the GPI}

Let us write the GPI in the form  
\begin{equation}
GPI_{n}=\delta (J_{n})  \label{sall1}
\end{equation}%
with  
\begin{equation}
J_{n}=\frac{1}{n}\overset{Q_{n}}{\underset{j=1}{\sum }}c(n,Q_{n},j)\text{ }%
d\left( \frac{Z-Y_{j,n}}{Z}\right) ,  \label{sall6}
\end{equation}%
where $c(n,Q_{n},j)=A(Q_{n},n,,Z)\times w(\mu _{1}n+\mu _{2}Q_{n}-\mu
_{3}j+\mu _{4})$ $/$ $B(Q_{n},n).$\ 

Since Y is an income or expenditure variable, its lower endpoint $y_{0}$\ is
not negative. This allows \ us to study (\ref{sall1}) via the transform $%
X=1/(Y-y_{0}).$\ Throughout this paper, the distribution function of $X$\ is
\begin{equation*}
F(\cdot )=1-G(y_{0}+1/\cdot ),
\end{equation*}%
whose upper endpoint is then +$\infty .$\ Hence (\ref{sall6}) is transformed
as  
\begin{equation}
J_{n}=\frac{1}{n}\overset{q}{\underset{j=1}{\sum }}c(n,q,j)\text{ }d\left( 
\frac{Z-y_{0}-X_{n-j+1,n}^{-1}}{Z}\right) .  \label{sall7}
\end{equation}

\bigskip\ 

We will need conditions on the function d($\cdot $) and on the weight
c(n,q,j), as in (\cite{lo2}). First assume that 

$\medskip $ 

$(D1)$\ $d(\cdot )$\ admits a continuous derivative on $]0,1)\medskip $. 

$(D2)$\ $d^{\prime }(\frac{z-y_{0}}{z})$\ and $d((z-y_{0})/z)$\ are
finite.\bigskip\ 

\qquad For $A(u)=1/F^{-1}(1-u)$, we assume that:\bigskip\ 

$(C1)$\ $A(\cdot )$\ is differentiable $(0,1)$\ ( and its derivative is
denoted $A^{\prime }(u)=a(u)$.) 

$(C2)$\ $a(\cdot )$\ is continuous on an interval $[a^{\prime },a^{\prime
\prime }]$\ with $0<a^{\prime }<a^{\prime \prime }<1$.\medskip\ 

$(C3)$\ $\exists \ u_{0}>0,\exists \ \eta >-3/2,\ \forall \ u\in \left(
0,u_{0}\right) ,\left\vert a(u)\right\vert <C_{0}\ u^{\eta }\exp
(\int_{u}^{1}b(t)t^{-1}dt),with$\ $b(t)\rightarrow 0\ as$\ $t\rightarrow 0.$ 

$\bigskip $ 

The condition (C3) means that $a(\cdot )$\ bounded by a regularly varying
function \  
\begin{equation*}
S(u)=C_{0}\ u^{\eta }\exp (\int_{u}^{1}b(t)t^{-1}dt)
\end{equation*}%
\ of exponent $\eta >-3/2$. \bigskip As to the function $\delta $, we need
it to be differentiable on $]0,+\infty \lbrack $, precisely :\newline

\bigskip $(E)$\ There is $\kappa >0$\ such that\ $\delta $($\cdot $) is
continuously differentiable on $]0,\kappa ]$.\newline
\bigskip \newline

We also need some conditions on the weight $c(\cdot )$. In order to state
the hypotheses, we introduce further notation. In fact we use in this paper
the representations of the studied random variables $X_{i},$\ $i\geq 1$,\ by 
$F^{-1}(1-U_{i}),$\ $i\geq 1$, where $U_{1},$\ $U_{2},...$\ is a sequence of
independent random variables uniformly distributed on $\left( 0,1\right) $.
Now let $U_{n}(\cdot )$\ and $V_{n}(\cdot )$\ be the uniform empirical
distribution and the empirical quantile function based on $U_{i},1\leq i\leq
n$.\ We have 

\begin{equation}
j\geq 1,\text{ \ }\frac{j-1}{n}<s\leq \frac{j}{n}\Longrightarrow \frac{j}{n}%
=U_{n}(V_{n}(s))  \label{sall10}
\end{equation}%
so that  
\begin{equation}
j\geq 1,\text{ \ }\frac{j-1}{n}<s\leq \frac{j}{n}\Longrightarrow
c(n,q,j)=c(n,q,nU_{n}(V_{n}(s))\equiv L_{n}(s).  \label{sall11}
\end{equation}%
Since $U_{n}(V_{n}(s))\rightarrow s,$\ as n$\rightarrow \infty ,$\ our
condition on the weight $c(\cdot )$\ is that the function $L_{n}(\cdot )$\
is uniformly bounded by some constant $D>0$\ and  
\begin{equation}
L_{n}(s)\rightarrow L(s),as\text{ }n\rightarrow \infty ,  \label{sall12}
\end{equation}%
where $L(\cdot )$\ is a non-negative $C^{1}-$function on $(0,1)$. 

$\bigskip $ 

We further require that, as $n\rightarrow \infty ,$\  
\begin{equation}
\underset{0\leq s\leq 1}{\sup }\left\vert \sqrt{n}(L_{n}(s)-L(s))-\gamma (s)%
\sqrt{n}(G_{n}(Z)-G(Z))\right\vert =o_{p}(1)  \label{sall12a}
\end{equation}%
for some function $\gamma (\cdot )$. Let us finally put  
\begin{equation*}
m(s)=L(s)\text{ }d\left( \frac{Z-y_{0}-1/F^{-1}(1-s)}{Z}\right) .
\end{equation*}

\bigskip We are now able to give our general theorem for the GPI. 

\begin{theorem}
Suppose that (C1-2-3), (D1-2) and (\ref{sall12a}) hold and let  
\begin{equation*}
\mu =\int_{0}^{G(Z)}\gamma (s)d\left( \frac{Z-y_{0}-1/F^{-1}(1-s)}{Z}\right) 
\text{ }ds.
\end{equation*}%
and  
\begin{equation*}
D=\int_{0}^{G(Z)}L(s)\text{ }d\left( \frac{Z-y_{0}-1/F^{-1}(1-s)}{Z}\right)
ds,
\end{equation*}%
Then  
\begin{equation*}
{\large \sqrt{n}(J_{n}-D)\rightarrow \mathcal{N}(0,\vartheta ^{2})}
\end{equation*}%
\ with  
\begin{equation*}
\vartheta ^{2}=\theta ^{2}+(m(G(Z))+\mu )^{2}G(Z)(1-G(Z))+\frac{%
2(m(G(Z))+\mu )}{Z}\int_{0}^{G(Z)}sL(s)h(s)ds
\end{equation*}%
and with  
\begin{equation*}
\theta ^{2}=Z^{-2}\int_{0}^{G(Z)}\int_{0}^{G(Z)}L(u)\text{ }L(v)\text{ }%
h(u)h(v)(u\wedge v-uv)\text{ }du\text{ }dv
\end{equation*}%
where  
\begin{equation*}
h(s)=a(s)\text{ }d^{\prime }(\frac{Z-y_{0}-1/F^{-1}(1-s)}{Z}).
\end{equation*}%
\bigskip If furthermore (E) holds and $D\in ]0,\kappa \lbrack ,$\ then  
\begin{equation*}
{\large \sqrt{n}(GPI}_{n}{\large -\delta (D))\rightarrow }\mathcal{N}{\large %
(0,\vartheta }^{2}\delta ^{\prime }(D)^{2})
\end{equation*}
\end{theorem}

$\bigskip $ 

The interest of this paper resides on the particular applications of the
theorem for the known indices. Before this, we give the guidelines of the
proof. 

\section{PROOFS OF THE RESULTS}

All our results will be derived from the lemma below. But, first we place
ourselves on a probability space where one version of the so-called
Hungarian constructions holds. Namely, M. Cs\"{o}rg\H{o} and al. (see \cite%
{cchm}) have constructed a probability space holding a sequence of
independent uniform random variables $U_{1},\ U_{2},$\ ... and a sequence of
Brownian bridges $B_{1},B_{2},...$\ such that for each $0<\nu <1/2$, $as$\ $%
n\rightarrow \infty ,$%
\begin{equation}
\underset{1/n\leq s\leq 1-1/n}{\sup }\frac{\left\vert \beta
_{n}(s)-B_{n}(s)\right\vert }{(s(1-s))^{1/2-\nu }}=O_{p}(n^{-\nu })
\label{sall12c}
\end{equation}%
and for each $0<\nu <1/4$ 
\begin{equation}
\underset{1/n\leq s\leq 1-1/n}{\sup }\frac{\left\vert B_{n}(s)-\alpha
_{n}(s)\right\vert }{(s(1-s))^{1/2-\nu }}=O_{p}(n^{-\nu }),  \label{sall12d}
\end{equation}%
where $\left\{ \alpha _{n}(s)=\sqrt{n}\left( U_{n}(s)-s)\right) ,0\leqslant
s\leqslant 1\right\} $\ is the uniform empirical process and $\left\{ \beta
_{n}(s)=\sqrt{n}\left( s-V_{n}(s)\right) ,0\leqslant s\leqslant 1\right\} $\
is the uniform quantile process. (See also \cite{ch} for a shorter and more
direct proof, and \cite{mvz} for dual version, in the sens that, \ref%
{sall12c} holds for $0<\nu <1/2$ and \ref{sall12d} for $0<\nu <1/4$ in \cite%
{ch}, while \ref{sall12c} is established for $0<\nu <1/4$ and \ref{sall12d}
for $0<\nu <1/2$ in \cite{mvz}). Throughout $\nu $\ will be fixed with $0<\nu
<1/4.$\ Now we are able to give the lemma. 

\begin{lemma}
Suppose that (C1-2-3) and (D1-2) hold and  
\begin{equation}
\underset{0\leq s\leq 1}{\sup }\sqrt{n}\left\vert L_{n}(s)-L(s)\right\vert
=O_{P}(1)\text{ as n}\rightarrow \infty .  \label{sall12b}
\end{equation}%
Let  
\begin{equation*}
D=\int_{0}^{G(Z)}L(s)\text{ }d\left( \frac{Z-y_{0}-1/F^{-1}(1-s)}{Z}\right)
ds.
\end{equation*}%
Then we have the expansion  

{\Large 
\begin{equation*}
\sqrt{n}(J_{n}-D)=N_{n}(1)+N_{n}(2)
\end{equation*}%
}

{\Large 
\begin{equation*}
+\int_{1/n}^{G(Z)}\sqrt{n}(L_{n}(s)-L(s))d\left( \frac{Z-y_{0}-1/F^{-1}(1-s))%
}{Z}\right) ds+o_{P}(1)
\end{equation*}%
}

\noindent with  
\begin{equation}
N_{n}(1)=\frac{1}{Z}\int_{1/n}^{G(Z)}L(s)B_{n}(s)h(s)ds  \label{salln1}
\end{equation}%
and  
\begin{equation}
N_{n}(2)=m(G(Z))B_{n}(G(Z))  \label{salln2}
\end{equation}%
for  
\begin{equation*}
m(s)=L(s)\text{ }d\left( \frac{Z-y_{0}-1/F^{-1}(1-s)}{Z}\right) .
\end{equation*}
\end{lemma}

\begin{proof}
This expansion is formulae (4.14) in (\cite{lo2}). Then, we have the
expansion  
\begin{equation*}
\sqrt{n}(J_{n}-C_{n})=\frac{1}{Z}\int_{1/n}^{G(Z)}L(s)B_{n}(s)h(s)\text{ }%
ds+n^{-1/2}L_{n}(1/n)\text{ }d\left( \frac{Z-y_{0}-1/F^{-1}(1-U_{1,n})}{Z}%
\right) 
\end{equation*}%
\begin{equation*}
+\int_{1/n}^{G_{n}(Z)}\sqrt{n}(L_{n}(s)-L(s))\text{ }d\left( \frac{%
Z-y_{0}-1/F^{-1}(1-V_{n}(s))}{Z}\right) ds
\end{equation*}%
\begin{equation*}
+\frac{1}{Z}\int_{G(Z)}^{G_{n}(Z)}L(s)B_{n}(s)h(s)ds+\frac{1}{Z}%
\int_{1/n}^{G_{n}(Z)}L_{n}(s)B_{n}(s)\text{ }(h(\zeta _{n}(s))-h(s))\text{ }%
ds
\end{equation*}%
\  
\begin{equation*}
+\frac{1}{Z}\int_{1/n}^{G_{n}(Z)}L_{n}(s)\text{ }(\beta _{n}\left( s\right)
-B_{n}(s))\text{ }h(\zeta _{n}(s))\text{ }ds
\end{equation*}%
It is proved in (\cite{lo2}) that  
\begin{equation*}
\sqrt{n}(J_{n}-C_{n})=N_{n}(1)+N_{n}(2)
\end{equation*}%
\begin{equation*}
+\int_{1/n}^{G_{n}(Z)}\sqrt{n}(L_{n}(s)-L(s))\text{ }d\left( \frac{%
Z-y_{0}-1/F^{-1}(1-V_{n}(s))}{Z}\right) ds+o_{P}(1).
\end{equation*}%
This gives  
\begin{equation*}
\sqrt{n}(J_{n}-C_{n})=N_{n}(1)+N_{n}(2)
\end{equation*}%
\begin{equation*}
+\int_{1/n}^{G(Z)}\sqrt{n}(L_{n}(s)-L(s))d\left( \frac{Z-y_{0}-1/F^{-1}(1-s))%
}{Z}\right) ds
\end{equation*}%
\begin{equation*}
+\int_{Gn(Z)}^{G(Z)}\sqrt{n}(L_{n}(s)-L(s))\text{ }d\left( \frac{%
Z-y_{0}-1/F^{-1}(1-V_{n}(s))}{Z}\right) ds+o_{P}(1)
\end{equation*}%
The condition (\ref{sall12b}) leads to the result. 
\end{proof}

$\bigskip $ 

We are now able to prove the Theorem. 

\begin{proof}
$\ $Let($\Omega $,$\Sigma $,P)\ be the probability space on which (\ref%
{sall12c}) and (\ref{sall12d}) hold. The Lemma together with (\ref{sall12b}%
), (\ref{sall12a}) and (\ref{salln2}), imply  
\begin{equation*}
{\large \sqrt{n}(J_{n}-D)=N}_{n}(1)+N_{n}(3)+o_{P}(1),
\end{equation*}%
where $N_{n}(1)$\ is defined in (\ref{salln1}) and  
\begin{equation}
N_{n}(3)=(m(G(Z)+\mu )\alpha _{n}(G(Z))+o_{P}(1)=(m(G(Z)+\mu
)B_{n}(G(Z))+o_{P}(1).  \label{sall21}
\end{equation}%
The vector $(N_{n}(1),N_{n}(3))$\ is Gaussian and  
\begin{equation}
cov(N_{n}(1),N_{n}(3))=\frac{m(G(Z))+\mu }{Z}%
E\int_{1/n}^{G(Z)}L(s)h(s)B_{n}(G(Z))B_{n}(s)ds  \label{salln4}
\end{equation}%
\begin{equation*}
=\frac{m(G(Z))+\mu }{Z}\int_{1/n}^{G(Z)}s\text{ }L(s)\text{ }h(s)\text{ }ds.
\end{equation*}%
Then $\sqrt{n}(J_{n}-D)$\ is a linear transform $N_{n}(1)+N_{n}(3)$\ of the
Gaussian vector $(N_{n}(1),N_{n}(3)),$\ plus an $o_{P}(1)$\ term. The
variance of this Gaussian term is easily computed through (\ref{salln4}) and
the conclusion follows, that is $\sqrt{n}(J_{n}-D)$\ is asymptotically a
centered Gaussian random variable with variance (\ref{salln4}). 
\end{proof}

\section{$\protect\bigskip $Asymptotic normality of particular indices}

\subsection{The FGT-like class}

This concerns the indices of the form  
\begin{equation*}
FGT(\alpha )=\frac{1}{n}\overset{Q_{n}}{\underset{j=1}{\sum }}d\left( \frac{%
Z-Y_{j,n}}{Z}\right) .
\end{equation*}%
We have here  
\begin{equation*}
L_{n}=1
\end{equation*}%
so that  
\begin{equation*}
\gamma =0
\end{equation*}%
Then  
\begin{equation*}
{\large \sqrt{n}(J}_{n}{\large -D)\rightarrow \mathcal{N}(0,\vartheta }^{2})
\end{equation*}%
\ with  
\begin{equation*}
\vartheta ^{2}=\theta ^{2}+m(G(Z)^{2}G(Z)(1-G(Z))+\frac{2m(G(Z))}{Z}%
\int_{0}^{G(Z)}sh(s)ds
\end{equation*}%
and  
\begin{equation*}
D=\int_{0}^{G(Z)}\left( \frac{Z-y_{0}-1/F^{-1}(1-s)}{Z}\right) ^{\alpha }ds.
\end{equation*}%
We should remark that the conditions $(D1-D2)$\ hold for $d(u)=u^{\alpha
},\alpha \geq 0.$ 

\subsubsection{The statistics nearby the FGT-class}

This concerns the statistics of the form  
\begin{equation*}
J_{n}=\delta \biggl(\frac{A(Q_{n},n)}{n}\overset{Q_{n}}{\underset{j=1}{\sum }}%
d\left( \frac{Z-Y_{j,n}}{Z}\right) \biggl) ,
\end{equation*}%
where we have a random weight not depending on the rank's statistic. We will
have two sub-cases. 

\subsubsection{The case of CHU's index}

Recall 

\begin{equation*}
CHU_{n}(\alpha )=\frac{Q_{n}}{nZ}\left\{ \frac{1}{Q_{n}}%
\sum_{j=1}^{Q_{n}}(Z-Y_{j,n})^{\alpha }\right\} ^{1/\alpha }
\end{equation*}%
\begin{equation*}
=\left\{ \frac{1}{n}\frac{Q_{n}^{\alpha -1}}{n^{\alpha -1}}%
\sum_{j=1}^{Q_{n}}(\frac{Z-Y_{j,n}}{Z})^{\alpha }\right\} ^{1/\alpha
}=\delta (J_{n})
\end{equation*}%
We easily get,  
\begin{equation*}
\sqrt{n}((q/n)^{\alpha -1}-G(Z)^{\alpha -1})=(\alpha -1)G(Z)^{\alpha -2}%
\sqrt{n}(G_{n}(Z)-G(Z))+o_{p}(1)
\end{equation*}%
\begin{equation*}
=(\alpha -1)G(Z)^{\alpha -2}B_{n}(G(Z))+o_{p}(1).
\end{equation*}%
By putting  
\begin{equation*}
C_{n}=FGT(\alpha )=\frac{1}{n}\sum_{j=1}^{Q_{n}}(\frac{Z-Y_{j,n}}{Z}%
)^{\alpha }
\end{equation*}%
and  
\begin{equation}
C=\int_{0}^{G(Z)}\left( \frac{Z-y_{0}-1/F^{-1}(1-s)}{Z}\right) ^{\alpha }ds,
\label{chu00}
\end{equation}%
we have, by the general theorem  
\begin{equation*}
\sqrt{n}(C_{n}-C)=N_{n}(1)+N_{n}(2)+o_{p}(1)
\end{equation*}%
with $L=1$\ . By combining these formulae, we get  
\begin{equation*}
\sqrt{n}(J_{n}-G(Z)^{\alpha -1}C)\rightarrow N(0,\zeta ^{2})
\end{equation*}%
with  
\begin{equation*}
\zeta ^{2}=\theta ^{2}+H(1-G(Z))\int_{0}^{G(Z)}s\text{ }%
a(s)ds+H^{2}G(Z)(1-G(Z))/2
\end{equation*}%
where,  
\begin{equation*}
H=C(\alpha -1)+G(Z)m(G(Z))G(Z)^{\alpha -2}.
\end{equation*}%
Finally, we get  
\begin{equation*}
\sqrt{n}(CHU_{n}(\alpha )-\delta (G(Z)^{\alpha -1}C)\rightarrow N(0,(\zeta
\delta ^{\prime }(G(Z)^{\alpha -1}C)^{2}),
\end{equation*}%
where  
\begin{equation*}
\delta ^{\prime }(G(Z)^{\alpha -1}C)^{2}=G(Z)^{-(\alpha -1)^{2}/\alpha
}C^{(1-\alpha )/\alpha }.
\end{equation*}

\subsubsection{The case of Ray's index}

Recall 

\begin{equation}
P_{R,n}(\alpha)=\frac{g}{nZ}\sum_{j=1}^{Q_{n}}((Z-Y_{j,n})/g)^{\alpha }
\end{equation}%
where  
\begin{equation}
g=\frac{1}{q}\sum_{j=1}^{j=Q_{n}}(Z-Y_{j,n}).
\end{equation}%
We have  
\begin{equation*}
J_{n}=g^{\alpha -1}\times C_{n}
\end{equation*}%
with  
\begin{equation*}
C_{n}=FGT_{n}(\alpha )
\end{equation*}%
and  
\begin{equation*}
C(\alpha )=\int_{0}^{G(Z)}\left( \frac{Z-y_{0}-1/F^{-1}(1-s)}{Z}\right)
^{\alpha }ds.
\end{equation*}%
We use the notation for the CHU index and we also get (\ref{chu00}). But  
\begin{equation*}
g=\frac{Zn}{Q_{n}}\times \frac{1}{n}\sum_{j=1}^{j=Q_{n}}(\frac{Z-Y_{j,n}}{Z}%
) \equiv \frac{Zn}{Q_{n}}K_{n}.
\end{equation*}%
We also have  
\begin{equation*}
\sqrt{n}(K_{n}-K)=\frac{1}{Z}%
\int_{1/n}^{G(Z)}B_{n}(s)a(s)ds+m_{1}(G(Z))B_{n}(G(Z))+o_{p}(1)
\end{equation*}%
with  
\begin{equation*}
K=C(1)=\int_{0}^{G(Z)}\frac{Z-y_{0}-1/F^{-1}(1-s)}{Z}ds
\end{equation*}%
and  
\begin{equation*}
\sqrt{n}(\frac{Zn}{q}-ZG(Z)^{-1})=Z\sqrt{N}(G(Z)-G_{n}(Z))G(Z)^{-2}+o_{p}(1)
\end{equation*}%
\begin{equation*}
=-Z\sqrt{n}B_{n}(G_{n}(Z))G(Z)^{-2}+o_{p}(1).
\end{equation*}%
By combining all that precedes, we arrive at  
\begin{equation*}
\sqrt{n}(g-KZG(Z)^{-1})=(m_{1}(G)ZG(Z)^{-1}-KZG(Z)^{-2})B_{n}(G(Z))
\end{equation*}%
\begin{equation*}
+\frac{1}{G(Z)}\int_{1/n}^{G(Z)}B_{n}(s)a(s)ds+o_{p}(1)
\end{equation*}%
and  
\begin{equation*}
\sqrt{n}(g^{\alpha -1}-(KZ/G(Z))^{\alpha -1})=(\alpha -1)(KZ/G(Z))^{\alpha
-2}
\end{equation*}%
\begin{equation*}
=(\alpha -1)(KZ/G(Z))^{\alpha -2}\text{ }%
(m_{1}(G)ZG(Z)^{-1}-KZG(Z)^{-2})B_{n}(G(Z))
\end{equation*}%
\begin{equation*}
+\frac{(\alpha -1)(KZ/G(Z))^{\alpha -2}}{G(Z)}%
\int_{1/n}^{G(Z)}B_{n}(s)a(s)ds+o_{p}(1)
\end{equation*}%
\begin{equation*}
=R_{1}B_{n}(G(Z))+R_{2}\int_{1/n}^{G(Z)}B_{n}(s)a(s)ds+o_{p}(1).
\end{equation*}%
Finally  
\begin{equation*}
\sqrt{n}(R_{n}-(KZ/G(Z))^{\alpha -1})C)=
\end{equation*}%
\begin{equation*}
\frac{(KZ/G(Z))^{\alpha -1})}{Z}\int_{1/n}^{G(Z)}B_{n}(s)h(s)ds+(KZ/G(Z))^{%
\alpha -1})\text{ }m_{\alpha }(G(Z))B_{n}(G(Z))
\end{equation*}%
\begin{equation*}
+CR_{1}B_{n}(G(Z))+CR_{2}\int_{1/n}^{G(Z)}B_{n}(s)a(s)ds+o_{p}(1)
\end{equation*}%
\begin{equation*}
=\int_{1/n}^{G(Z)}B_{n}(s)\psi (s)ds+\left\{ (KZ/G(Z))^{\alpha -1})\text{ }%
m_{\alpha }(G(Z))+CR_{1}\right\} B_{n}(G(Z))+o_{P}(1)
\end{equation*}

\begin{equation*}
=\int_{1/n}^{G(Z)}B_{n}(s)a(s)ds+A_{2}B_{n}(G(Z))+o_{P}(1),
\end{equation*}%
with  
\begin{equation*}
\psi (s)=a(s)\left\{ C(\alpha )R_{2}+(KZ/G(Z))^{\alpha -1}Z^{-1}d^{\prime
}(Z^{-1}(Z-y_{0}-1/F^{-1}(1-s)))\right\} .
\end{equation*}%
Notice that $\int_{1/n}^{G(Z)}B_{n}(s)h(s)ds+A_{1}B_{n}(G(Z))$\ is a normal
centered random variable with variance  

{\Large 
\begin{multline*}
\xi ^{2}=\int_{0}^{G(Z)}\int_{0}^{G(Z)}\psi (u)\psi (v)(u\wedge v-uv)\text{ }%
du\text{ }dv \\
+A_{1}^{2}G(Z)(1-G(Z))+2A_{1}(1-G(Z))\int_{0}^{G(Z)}s\text{ }\psi (s)ds.
\end{multline*}%
}

\noindent We conclude that  
\begin{equation*}
\sqrt{n}(P_{R,n}(\alpha)-(KZ/G(Z))^{\alpha -1})C)\rightarrow _{d}N(0,\xi ^{2})
\end{equation*}%
with  
\begin{equation*}
m_{\alpha }(u)=(Z^{-1}(Z-y_{0}-1/F^{-1}(1-s))^{\alpha },
\end{equation*}%
\begin{equation*}
R_{1}=(\alpha -1)(KZ/G(Z))^{\alpha -2}\text{ }%
(m_{1}(G)ZG(Z)^{-1}-KZG(Z)^{-2}),
\end{equation*}%
\begin{equation*}
R_{2}=(\alpha -1)(KZ/G(Z))^{\alpha -2}G(Z)^{-1},
\end{equation*}%
and  
\begin{equation*}
A_{1}=(KZ/G(Z))^{\alpha -1})\text{ }m_{\alpha }(G(Z))+CR_{1}
\end{equation*}

\subsection{The Shorrocks-like indices}

This concerns the Thon and Shorrocks measures. They both have a similar
asymptotic behavior. 

For Shorrocks's index, we have  
\begin{equation*}
P_{SH,n}=\frac{1}{n^{2}}\overset{Q_{n}}{\underset{j=1}{\sum }}%
(2n-2j+1)\left( \frac{Z-Y_{j,n}}{Z}\right) .
\end{equation*}%
But  
\begin{equation}
j\geq 1,\ \frac{j-1}{n}<s\leq \frac{j}{n}\Longrightarrow
L_{n}(s)=c(n,q,j)=(2-2\ast j/n+1/n)  \label{sall14}
\end{equation}%
\begin{equation*}
\rightarrow L(s)=2(1-s),
\end{equation*}%
and,  
\begin{equation*}
\sqrt{n}(L_{n}(s)-L(s))=-2\ast \sqrt{n}(U_{n}(V_{n}(s))-s)+1/\sqrt{n,}
\end{equation*}%
By\ (\cite{shwell}), $\ $p.151,  
\begin{equation*}
\sqrt{n}\underset{0\leq s\leq 1}{\sup }\left\vert L_{n}(s)-L(s)\right\vert
\leq 3/\sqrt{n}.
\end{equation*}%
and then  
\begin{equation*}
\gamma \equiv 0\text{, }h(\cdot )=a(\cdot )
\end{equation*}%
For the Thon Statistic,  
\begin{equation*}
P_{T,n}=\frac{2}{n(n+1)}\overset{Q_{n}}{\underset{j=1}{\sum }}(n-j+1)\left( 
\frac{Z-Y_{j,n}}{Z}\right) ,
\end{equation*}%
we also have  
\begin{equation*}
L(s)=2(1-s),\text{ }\gamma \equiv 0,\text{ }a(\cdot )=h(s).
\end{equation*}%
In both cases, for $P_n=P_{SH,n}$ or $j_n=P_{T,n}$, we get  
\begin{equation*}
{\large \sqrt{n}(P}_{n}{\large -D)\rightarrow \mathcal{N}(0,\vartheta }^{2})
\end{equation*}%
\ with  
\begin{equation*}
D=2\int_{0}^{G(Z)}(1-s)\text{ }\left( \frac{Z-y_{0}-1/F^{-1}(1-s)}{Z}\right)
ds,
\end{equation*}%
\begin{equation*}
\vartheta ^{2}=\theta ^{2}+m(G(Z)G(Z)(1-G(Z))+\frac{4m(G(Z))}{Z}%
\int_{0}^{G(Z)}s(1-s)a(s)ds
\end{equation*}%
and with  
\begin{equation*}
\theta ^{2}=4Z^{-2}\int_{0}^{G(Z)}\int_{0}^{G(Z)}(1-u)(1-v)\text{ }%
a(u)a(v)(u\wedge v-uv)\text{ }du\text{ }dv.
\end{equation*}

\subsection{The Kakwani-class.}

The Kakwani class  
\begin{equation*}
P_{KAK,n}=\frac{Q_{n}}{n\Phi _{k}(Q_{n})}\overset{Q_{n}}{\underset{j=1}{\sum 
}}(Q_{n}-j+1)^{k}\left( \frac{Z-Y_{j,n}}{Z}\right) ,
\end{equation*}%
is introduced with a positive integer. We consider here that k is merely a
non-negative real number. It is proved in (\cite{lo3}) that  
\begin{equation*}
L(s)=(k+1)(1-s/G(Z))^{k}
\end{equation*}%
and that  
\begin{equation*}
\gamma (s)=k(k+1)(1-s/G(Z))^{k-1}(s/G(Z)^{2}).{\large \ }
\end{equation*}%
We remark that $m(G(Z))=0.$\ Then our result is particularized as  
\begin{equation*}
{\large \sqrt{n}(P}_{KAK,n}(k){\large -D)\rightarrow \mathcal{N}(0,\vartheta 
}^{2})
\end{equation*}%
\ with  
\begin{equation}
\vartheta ^{2}=\theta ^{2}+\mu ^{2}G(Z)(1-G(Z))+\frac{2\mu }{Z}%
(1-G(Z))\int_{0}^{G(Z)}sL(s)h(s)ds  \label{salln3}
\end{equation}%
and with  
\begin{equation*}
\theta ^{2}=Z^{-2}\int_{0}^{G(Z)}\int_{0}^{G(Z)}L(u)\text{ }L(v)\text{ }%
h(u)h(v)(u\wedge v-uv)\text{ }du\text{ }dv.
\end{equation*}%
for a fixed real number $k\geq 1.$ 

\bigskip\ 

We have now finished the poverty indices' review. Some of these results have
been simulated and applied in particular issues with the Senegalese Data. 

\section{Conclusion}

The GPI includes most of the poverty indices. We have established here their
asymptotic normality with immediate applications to poor countries data for
finding accurate confidence intervals of the real poverty measurement. In
coming papers, a special study will be devoted to the Takayama statistic.
The GPI is to be thoroughly visited through the poverty axiomatic approach
as well.


\newpage 

\begin{thebibliography}{99}
\bibitem{chakravarty} Chakravarty, S. R.(1983). A new Poverty Index.
Mathematical Social Science 6, 307-313. 

\bibitem{clark} Clark, S., Hemming, R. and Ulph, D.(1981). On Indices for
the Measurement of Poverty. Economic Journal 91, 525-526. 

\bibitem{ch} Cs\"{o}rg\H{o}, M. Horv\`{a}th, L.(1986). Approximations of
weighted empirical and quantile processes. Statistics and Probability
letters, 4, 275-280. 

\bibitem{cchm} Cs\"{o}rg\H{o}, M., Cs\"{o}rg\H{o} S., Horv\`{a}th, L. and
Mason, M. (1986). Weighted Empirical and Quantile Processes. Ann. Probab.
14,31-85. 

\bibitem{dia} Dia, G.(2005). R\'{e}partition Ponctuelle Al\'{e}atoire des
Revenus et Estimation de l'Indice de Pauvret\'{e}. Afrik. Stat., Vol. 1 (1),
p.47-66. 

\bibitem{fgt} Foster, J., Greer, J. and Shorrocks, A.(1984). A class of
Decomposable Poverty Measures, Econometrica 52, 761-766. 

\bibitem{kakwani} Kakwani, N.(1980). On a Class of Poverty Measures.
Econometrica, 48, 437-446. 

\bibitem{kak2} Kakwani, N. (2003). Issues on Setting Absolute Poverty Line.
Poverty and Social Development, Papers, $n^{o}$3, Asian Bank of
Development Bank (ADB). 

\bibitem{lo2} Lo, G.S., SALL, S.T., Seck C.T.(2009) Une Th\'eorie
asymptotique des indicateurs de pauvret\'e. \textit{C. R. Math. Rep. Acad.
Sci. Canada}, 45-52, 31 (2).

\bibitem{lo6} Lo, G. S. (2008). Estimation asymptotique des indices de
pauvret\'{e} : mod\'{e}lisation continue et analyse spatio-temporelle de la
pauvret\'{e} au S\'{e}n\'{e}gal. \textit{Journal Africain des Sciences de la Communication et des Technologies}, 341-377, D\'ecembre, Vol. 4.\\

Disponible \`{a} : $http://www.ufrsat.org/lerstad/pub/gslo_{-}gpi_{-}arxiv.pdf$

\bibitem{lo1} Sall, S. T. and Lo, G.S., (2007). The Asymptotic Theory of
Intensity Poverty in View of Extreme Values Theory For Two Simples Cases.
Afrika Statistika, vol 2, $n^{o}$, 1, p.41-55 

\bibitem{mvz} Mason D. M. and van Zwet, W.R.(1987). A refinement of the KMT
inequality for the uniform empirical process. Ann. Probab., 15(3), 871--884. 

\bibitem{sen} Sen Amartya K.(1976). Poverty: An Ordinal Approach to
Measurement. Econometrica,\ 44, 219-231. 

\bibitem{shwell} Shorack G.R. and Wellner J. A.(1986). Empirical Processes
with Applications to Statistics, wiley-Interscience, New-York. 

\bibitem{shorrocks} Shorrocks, A.(1995). Revisiting the Sen Poverty Index.
Econometrica 63, 1225-1230. 

\bibitem{takayama} Takayama, N.(1979). Poverty, Income Inequality, and Their
Measures : Professor Sen's Axiomatic Approach Reconsidered, Econometrica 47,
747-759. 

\bibitem{thon} Thon, D.(1979). On Measuring Poverty. Review of Income and
Wealth 25, 429-440. 

\bibitem{watts} Watts, H.(1968). An Economics Definition of Poverty. In D.
P. Moyniha (ed), On Understanding Poverty, New-York : Basic Books. 

\bibitem{zheng} Zheng, B.(1997). Aggregate Poverty Measures. Journal of
Economic Surveys 11 (2), 123-162. 

\bibitem{ray} Ray, R.(1989), A new class of decomposable poverty measures.
Indian Economy Journal, vol.36, 30-38. 

\bibitem{ravallion} Ravallion, M. (1992). Poverty Compairisons. A Guide to
Concepts and Methods. Lsms, Working Paper, $n^{o}$ 88,WorldBank. 
\end{thebibliography}
\end{document}